\documentclass[letterpaper,11pt]{scrartcl}
\usepackage[todo]{brandl}
\usepackage[margin=1in]{geometry}
\usepackage{enumitem}

\usepackage{array,xspace,hhline,colortbl,tabularx,booktabs,fixltx2e}
\usepackage{ifthen}
\usepackage{algorithm}
\usepackage{algorithmic}
\usepackage{adjustbox}
\usepackage{blkarray}
\usepackage{enumitem}
\usepackage{bm}
\usepackage[font=itshape]{quoting}

\theoremstyle{remark}

\newboolean{comments}
\setboolean{comments}{false} 

\ifthenelse{\boolean{comments}}{%
\newcommand{\haris}[1]{\textcolor{blue}{\textbf{Haris says:} #1}}

\newcommand{\fb}[1]{\textcolor{purple}{\textbf{Florian says:} #1}}
\newcommand{\fixme}[1]{{\color{red}\fbox{FIXME} #1}}
}{%
\newcommand{\haris}[1]{}

\newcommand{\fb}[1]{}
\newcommand{\fixme}[1]{}
}

	\newcommand{\rr}[0]{\ensuremath{\mathit{REV}}\xspace}
	\newcommand{\srr}[0]{\ensuremath{\mathit{S}}--\ensuremath{\mathit{REV}}\xspace}
	
	\newcommand{\midd}{\mathbin{:}}
	\renewcommand{\pref}{\succsim\xspace}
	\newcommand{\apref}[1][]{%
	\ifthenelse{\equal{#1}{}}{{\pref}_N}{{\pref^{#1}}_N}%
	}
	\newcommand{\aprefs}[1][]{%
	\ifthenelse{\equal{#1}{}}{{\succ}_N}{{\succ^{#1}}_N}%
	}
	\newcommand{\opref}[1][]{%
	\ifthenelse{\equal{#1}{}}{{\pref}_O}{{\pref^{#1}}_O}%
	}
	\newcommand{\oprefs}[1][]{%
	\ifthenelse{\equal{#1}{}}{{\succ}_O}{{\succ^{#1}}_O}%
	}
	
	\newcommand{\pss}[1][]{%
	\ifthenelse{\equal{#1}{}}{({\succ}_N,{\succ}_O)}{({\succ^{#1}}_N,{\succ^{#1}}_O)}%
	}
	
	\newcommand{\quotas}[1][]{%
	\ifthenelse{\equal{#1}{}}{(q_c)}{(q_c#1)}%
	}
	\newcommand{\prios}[1][]{%
	\ifthenelse{\equal{#1}{}}{({\succsim_c})}{({\succsim_c#1})}%
	}
	\newcommand{\ms}{\mathit{ms}}

	\newtheorem*{theorem*}{Theorem}

    \makeatletter
\def\@fnsymbol#1{\ensuremath{\ifcase#1\or \dagger\or \ddagger\or
   \mathsection\or \mathparagraph\or \|\or **\or \dagger\dagger
   \or \ddagger\ddagger \else\@ctrerr\fi}}
    \makeatother

\begin{document}

\title{Efficient, Fair, and Incentive-Compatible Healthcare Rationing}


\author{Haris Aziz\thanks{UNSW Sydney and Data61 CSIRO, Australia, \texttt{haris.aziz@unsw.edu.au}} \and Florian Brandl\thanks{University of Bonn, Germany, \texttt{florian.brandl@uni-bonn.de}}}
 

	\maketitle

\begin{abstract}
			Rationing of healthcare resources has emerged as an important issue, which has been discussed by medical experts, policy-makers, and the general public. We consider a rationing problem where medical units are to be allocated to patients. Each unit is reserved for one of several categories and each category has a priority ranking of the patients. We present an allocation rule that respects the priorities, complies with the eligibility requirements, allocates the largest feasible number of units, and does not incentivize agents to hide that they qualify through a category. The rule characterizes all possible allocations that satisfy the first three properties and is polynomial-time computable.

		\par\vskip\baselineskip
		\textbf{Keywords:} Allocation under priorities, healthcare rationing, assignment maximization.
\end{abstract}


\section{Introduction}

The COVID-19 pandemic has emerged as one of the major challenges the world has faced. It resulted in a frantic race to produce the most effective and safe vaccine to stem the devastating effects of the pandemic. 
While the creation, testing, and approval of vaccines proceeded at an unprecedented pace, the initial scarcity posed the question of how to distribute them efficiently and fairly.
The same problem arises for other scarce or costly healthcare resources such as ventilators and antiviral treatments.
One approach for allocating resources is to assign patients to priority groups, which capture how much a patient needs treatment.
For example, three priority groups that have been highlighted by medical practitioners and policy-makers are health care workers, other essential workers and people in high-transmission settings, and people with medical vulnerabilities associated with poorer COVID-19 outcomes~\citep{PPE20a,TMD20a}. Other concerns that have been discussed include racial equity~\citep{BrTl21a}.

It is however not enough to identify priority groups. 
There is also a need to algorithmically and transparently make allocation decisions~\citep{EPU+20a,WHO20a} based on the priorities. 
In a New York Times article, the issue has been referred to as one of the hardest decisions health organizations need to make~\citep{Fink20a}. 
Since the decisions need to be justified to the public, they must be aligned with ethical guidelines, such as respecting the priorities among the patients.
These decisions are not straightforward, especially when a patient is eligible for more than one category. 
In that case, the decision of which category is used to serve a patient can have compounding effects on what categories other agents can use.
A competing objective is allocating the resources efficiently in the sense that they have maximum social benefit. 
For example, no medical unit should be left unallocated.
Thus, the following fundamental question arises.
\begin{quoting}
How should we allocate scarce medical resources fairly and efficiently while taking into account various ethical principles and priority groups?
\end{quoting}


The problem of health care rationing has recently been formally studied by market designers. 
\citet*{PSU+20a} were among the first to frame the problem as a two-sided matching problem in which patients are on one side and the resources units are on the other side. 
In doing so, they linked the healthcare rationing problem with the rich field of two-sided matching~\citep{RoSo90a}. 
\citeauthor{PSU+20a} suggested dividing the units into different reserve categories, each with its own priority ranking of the patients. The categories and the category-specific priorities represent the ethical principles and guidelines that a policy-maker may wish to implement.\footnote{See for example, the book by \citet{BoHi14a} on the ethics of healthcare rationing that discusses many of these principles.} For example, a category for senior people may have an age-specific priority ranking that puts the eldest citizens first. Having a holistic framework that considers different types of priorities has been termed important in healthcare rationing.\footnote{In a report issued by the Deeble Institute for Health Policy Research, \citet{Mart15a} writes that ``To establish robust healthcare rationing in Australia, decision-makers need to acknowledge the various implicit and explicit priorities that influence the process and develop a decision-making tool that incorporates them.''} The approach of \citeauthor{PSU+20a} has been  recommended or adopted by various organizations including the
NASEM (National Academies of Sciences, Engineering, and Medicine) Framework for Equitable Vaccine Allocation~\citep{NASEM20a} and has been endorsed in medical circles~\citep{PPE20a,SPU+20b}.
The approach has been covered widely in the media, including the New York Times and the Washington Post.\footnote{\url{https://www.covid19reservesystem.org/media}}

For their two-sided matching formulation, \citet{PSU+20a} proposed a solution for the problem. 
One of their key insights was that  running the Deferred Acceptance algorithm~\citep{GaSh62a,Roth08a} on the underlying problem satisfies basic relevant axioms (eligibility compliance, respect of priorities, and non-wastefulness). 
They also showed when all the category priorities are consistent with a global baseline priority, their Smart Reserves algorithm computes a maximum size matching satisfying the basic axioms. 
The Smart Reserves algorithm makes careful decisions about which category should be availed by which patient to achieve the maximum size property. 
However, the general problem with general \emph{heterogeneous} priorities has not been addressed in the literature. 
Allowing heterogeneous priorities for categories seems to be very much in the spirit of incorporating different ethical values. 
For example, one would expect the priority ordering for old people to be very different from a priority ordering for front-line workers which would favor energetic medical professionals.\footnote{Even in other contexts such as immigration, where rationing policies are applied, respecting heterogeneous priorities can be important. For example, if a country has a quota for admitting engineers, the top engineering applicant who satisfies basic eligibility requirements may have a good case to be issued a visa.}
In this paper, we set out to address this issue and answer the following research problem. 
\begin{quoting}
For healthcare rationing with heterogeneous priorities, how should we allocate resources in a fair, economically efficient, strategyproof, and computationally tractable way?
\end{quoting}

\paragraph{Contribution}
We consider the healthcare rationing with heterogeneous priorities.
First we highlight that naively ascribing strict preferences over the categories to the agents
can have adverse effects on the efficiency of the outcome when patients are eligible for multiple categories. If the eligibility requirements are treated as hard constraints, it leads to inefficient allocation of resources. If the eligibility requirements are treated as soft constraints, then the outcome does not allocate the resources optimally to the highest priority patients,
thereby undermining important ethical principles.  

Our first contribution is introducing the Reverse Rejecting (\rr) rule, which
\begin{enumerate}[label=(\roman*)]
	\item complies with the eligibility requirements,
	\item respects the priorities of the categories (for each category, patients of higher priority are served first),
	\item yields a maximum size matching (one that allocates the largest feasible number of units to eligible patients),
		\item is non-wasteful (no unit is unused but could be used by some eligible patient),
	\item is strategyproof (does not incentivize agents to under-report the categories they qualify for), and
	\item is strongly polynomial-time computable.
\end{enumerate}

We prove that the \rr rule characterizes all outcomes that satisfy the first three properties.\footnote{In a preliminary version of the paper, we proposed a different rule that satisfies the above properties. It requires solving as maximum weight matching of a corresponding graph. 
 However, it does not characterize all outcomes that satisfy the first three properties.}
We show how the \rr rule can be extended to a more general rule called the \emph{Smart Reverse Rejecting (\srr)} rule, which processes given numbers of unreserved units first and last and which incorporates the additional goal of allocating the largest feasible number of units from a designated subset of categories called preferential categories. 
The \srr rule satisfies a new axiom we call order preservation, which is parameterized by how many unreserved units are processed first and last.  
Moreover, it generalizes two well-known reserves rules---over-and-above and minimum-guarantees~\citep{Gala61a,Gala84a}---that are understood in the context of preferential categories having consistent priorities. 
Finally, we discuss how our algorithms and their properties extend to the case where the reservations are treated as soft reservations.


Our algorithm immediately applies to the school choice problem in which students are only interested in being matched to one of their acceptable schools. It also applies to hiring settings in which applicants are interested in one of the positions, and each of the departments has its own priorities. 
Lastly, it applies to many other rationing scenarios, such as the allocation of limited slots at public events or visas to immigration applicants. 

\section{Related Work}

The paper is related to an active area of research on matching with distributional constraints \citep[see, e.g.,][]{Koji19a,ABY21a}. One general class of distributional constraints that have been examined in matching market design pertains to common quotas over institutions such as hospitals~~\citep{KaKo15a,KaKo17a,BFIM10a,GIKY+16a}. 

Within the umbrella of work on matching with distributional constraints, particularly relevant to healthcare rationing is the literature on school choice with diversity constraints and reserve systems~\citep{HYY13a,EHYY14a,EcYe15a,KHIY17a,AyTu20a,AyBo20a,AGS20a,GNKR19a,UKPS18a,DPS20a}. 
Categories in healthcare rationing correspond to affirmative action types in school choice.
For a brief survey, we suggest the book chapter by \citet{Heo19a}. Except for the special case in which students have exactly one type \citep[see, e.g.,][]{EHYY14a}, most of the approaches do not achieve diversity goals optimally, whereas for the healthcare rationing problem we consider, we aim to find matchings that maximize the number of units allocated to eligible patients. \citet{ADF17a}, \citet{DSSX19a}, and \citet{AAD+20a} 
consider optimisation approaches for diverse matchings, but their objective and models are different.

\citet{PSU+20a} were the first to frame a rationing problem with category priorities as a two-sided matching problem in which agents are simply interested in a unit of resource and the resources are reserved for different categories. They show that artificially enforcing strict preferences of the agents over the categories and running the Deferred Acceptance algorithm results in desirable outcomes for the rationing problem. They note, however, that this approach may lead to matchings that are not Pareto optimal. 
They then proposed to use the smart reserves approach of \citet{SoYe20a} for the restricted problem when all the {preferential} categories have {consistent} priorities.
Our results can be viewed as simultaneously achieving the  {key} axioms of the two approaches of \citet{PSU+20a}. Firstly, we propose a new algorithm that achieves the same key axiomatic properties for heterogeneous priorities as the algorithms of \citet{SoYe20a} and \citet{PSU+20a} for homogeneous priorities. Secondly, our algorithm has an important advantage over the Deferred Acceptance formulation of \citet{PSU+20a} for the case of heterogeneous priorities as our approach additionally achieves the important property of maximality of size. {It additionally satisfies a property called maximality in beneficiary assignment, which requires that the maximum number of units from the set of `preferential' categories are used. }
{\citet{PSU+20a} design a flexible feature of their Smart Reserves rule that gives agents a designated number of unreserved units before the other units are processed. 
By doing so, they elegantly capture two extreme approaches within their class that have wide-spread appeal. The first approach is based on \emph{minimum-guarantees} that specifies the minimum number of units that are kept for a particular agent group. 
The second approach is \emph{over-and-above}; it sets aside the specified number of units for an agent group and only uses them once all the unreserved units are allocated (for which the agent group is eligible as well).\footnote{Both the ``minimum-guarantees'' and ``over-and-above'' approaches have been discussed in the context of reserves systems in India (see, e.g., \citet{Gala61a,Gala84a}).}
Our \srr rule achieves these features even for the case of heterogeneous priorities.}
 In follow-up work, \citet{Grig20a} considers optimisation approaches for variants of the problem but does not present any polynomial-time algorithm or consider incentive issues. In contrast to the papers on healthcare rationing discussed above, we also consider strategyproofness aspects and show that our rule complies with them.

In this paper, we attempt to compute what are essentially maximum size stable matchings. The problem of computing such matchings is NP-hard if both sides have strict preferences/priorities~\citep{BMM10a}. In our problem, the agents essentially have dichotomous preferences (categories they are/are not eligible for) and, hence, we are able to obtain a polynomial-time algorithm for the problem. 

Furthermore, our rules are strategyproof. In contrast, for other two-sided matching settings, it is known that maximizing the number of matched individuals results in incentive and fairness impossibilities \citep[see, e.g.,][]{ABT20a,KMRZ+14a}. Computing outcomes that match as many agents as possible, has also been examined in related but different contexts \citep[see, e.g.,][]{Aziz17b,AnEh17a,ABS07a,BoMo15a}.

%

\section{Model}

We adopt the model of \citet{PSU+20a} with one generalization: we allow the categories' priorities over agents to be weak rather than strict. 
There are $q$ identical and indivisible units of some resource, a finite set $C$ of categories, and a set $N$ of agents with $|N| = n$.
Each category $c$ has a quota $q_c\in\mathbb N$ with $\sum_{c\in C}q_c=q$ and a priority ranking $\pref_c$, which is a preorder on $N \cup\{\emptyset\}$. 
An agent $i$ is \emph{eligible} for category $c$ if $i \succ_c\emptyset$. 
We denote by $N_c$ the set of agents who are eligible for $c$.
We say that $I = (N,C, \prios_{c\in C}, \quotas_{c\in C})$ is an instance (of the rationing problem).
We will write $\prios$ and $(q_c)$ for the profile of priorities and quotas in the sequel.
We also consider a baseline ordering $\succ_\pi$, which can be an arbitrary permutation of the agents.
It could be interpreted as a global scale measuring the need for treatment.

A matching $\mu\colon N\rightarrow C\cup \{\emptyset\} $ is a function that maps each agent to a category or to $\emptyset$ and satisfies the capacity constraints: for each $c\in C$, $|\mu^{-1}(c)|\leq q_c$.
For an agent $i\in N$, $\mu(i)=\emptyset$ means that $i$ is unmatched (that is, does not receive any unit) and $\mu(i)=c$ means that $i$ receives a unit reserved for category $c$.
When convenient, we will identify a matching $\mu$ with the set of agent-category pairs $\{\{i,\mu(i)\}\colon \mu(i)\neq\emptyset\}$.\footnote{In graph theoretic terms, $\mu$ is a $b$-matching because multiple edges in $\mu$ can be adjacent to a category $c$.}

We introduce four axioms in the context of allocating medical units that are well-grounded in practice. 
For further motivation of these axioms, we recommend the detailed discussions by \citet{PSU+20a}.

The first axiom requires matchings to comply with the eligibility requirements. 
It specifies that a patient should only take a unit of a category for which the patient is eligible. 
For example, a young person should not take a unit from the units reserved for elderly people. 

\begin{definition}[Compliance with eligibility requirements]\label{def:eligibility}
A matching $\mu$ \emph{complies with eligibility requirements} if for any $i\in N$ and $c\in C$, $\mu(i)=c$ implies $i \succ_c \emptyset$.
\end{definition}

The second axiom concerns the respect of priorities of categories. It rules out that a patient is matched to a unit of some category $c$ while some other agent with a higher priority for $c$ is unmatched.  

\begin{definition}[Respect of priorities]\label{def:priorities}
A matching $\mu$ \emph{respects priorities} if for any $i,j\in N$ and $c\in C$, $\mu(i)=c$ and $\mu(j)=\emptyset $ implies $ j\not\succ_c i$. 
If there exist $i,j\in N$ and $c\in C$ with $\mu(i)=c$, $\mu(j)=\emptyset$, and $ j \succ_c i$, we say that $j$ has \emph{justified envy} towards $i$ for category $c$.
\end{definition}

An astute reader who is familiar with the theory of stable matchings will immediately realize that the axiom ``respect of priorities'' is equivalent to \emph{justified envy-freeness} in the context of school-choice matchings~\citep{AbSo03a}.

Next, non-wastefulness requires that if an agent is unmatched despite being eligible for a category, then all units reserved for that category are matched to other agents.

\begin{definition}[Non-wastefulness]\label{def:nonwasteful}
A matching $\mu$ is \emph{non-wasteful} if for any $i\in N$ and $c\in C$, $i \succ_c \emptyset$ and $\mu(i)=\emptyset $ implies $ |\mu^{-1}(c)| =q_c$.
\end{definition}

Not all non-wasteful matchings allocate the same number of units. 
In particular, some may not allocate as many units as possible.
A stronger efficiency notion prescribes that the number of allocated units is maximal subject to compliance with the eligibility requirements.

\begin{definition}[Maximum size matching]\label{def:maximum}
A matching $\mu$ is a \emph{maximum size matching} if it has maximum size among all matchings complying with the eligibility requirements.
\end{definition}

These four axioms capture the first guideline put forth in the report by the National Academies of Sciences, Engineering, and Medicine: ``ensure that allocation maximizes benefit to patients, mitigates inequities and disparities, and adheres to ethical principles''~\citep[page 69]{NASEM20a}. {Requiring matchings to be of maximum size is aligned with the principle to ``gain the best value we possibly can from the expenditure of that resource''~\citep{DIJ+20a}.}

It will be useful to associate a graph $B_I$, called a \emph{reservation graph}, with an instance $I$.
$B_I = (N\cup C,E)$ is a bipartite graph with an edge from $i$ to $c$ if $i$ is eligible for $c$.
That is, $E = \{\{i,c\}\colon i\succ_c\emptyset\}$.
If $G$ is any graph, we denote by $\ms(G)$ the number of edges in a maximum size matching of $G$.

The following example illustrates the definitions above.
\begin{example}\label{ex:running}
	Suppose there are three agents and two categories with one reserved unit each. 
	\[N=\{1,2,3\},\quad C=\{c_1,c_2\},\quad q_{c_1}=1, q_{c_2}=1.\]
The priority ranking of $c_1$ is $2\succ_{c_1} 3 \succ_{c_1} \emptyset \succ_{c_1} 1$ and the priority ranking of $c_2$ is $2\succ_{c_2} \emptyset  \succ_{c_2} 1 \succ_{c_2} 3$.
\Cref{fig:firstexample-0} illustrates this instance $I$ of the rationing problem. 
	 
	Note that agent 1 is not eligible for any category, agent 2 is eligible for $c_1$ and $c_2$, and agent 3 is eligible only for $c_1$. 	
Thus, the following matchings comply with the eligibility requirements.
\begin{align*}
	&\mu_1 = \emptyset &&\mu_2 = \{\{2,c_1\}\} && &\mu_3 = \{\{2,c_2\}\}\\
	&\mu_4 = \{\{3,c_1\}\} &&\mu_5 = \{\{2,c_2\},\{3,c_1\}\}
\end{align*}
All of these matchings except for $\mu_4$ respect priorities.
Only $\mu_2$ and $\mu_5$ are non-wasteful.
The only maximum size matching is $\mu_5$.
\end{example}
	
	   	   	    	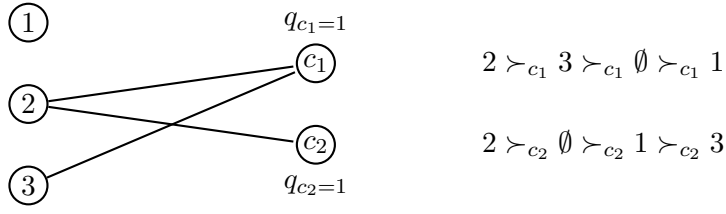
\begin{figure}[]
	   	   	    		\begin{center}
	   	   	    	\begin{tikzpicture}[-,>=stealth',shorten >=1pt,auto,node distance=3cm,
	   	   	    	        thick,main node/.style={circle,fill=white!20,draw,minimum size=0.5cm,inner sep=0pt]}, scale=0.9]
			\def\dist{1.2}

	   	   	    	    \node[main node] at (0, 0) (1)  {$1$};
	   	   	    	    \node[main node] at (0, -\dist) (2)  {$2$};
	   	   	    	    \node[main node] at (0, -2*\dist) (3)  {$3$};

	   	   	    	    \node[main node] at (3.5*\dist, -.5*\dist) (c1)  {$c_1$};
	   	   	    	    \node[main node] at (3.5*\dist, -1.5*\dist) (c2)  {$c_2$};
						
								\node[] at (3.5*\dist, -2*\dist) {{$q_{{c_2}=1}$}};
						
		\node[] at (7*\dist, -.5*\dist) {{$2\succ_{c_1} 3 \succ_{c_1} \emptyset \succ_{c_1} 1$}};
		
		\node[] at (3.5*\dist, -0) {{$q_{{c_1}=1}$}};
		
		\node[] at (7*\dist, -1.5*\dist) {{$2\succ_{c_2} \emptyset  \succ_{c_2} 1 \succ_{c_2} 3$}};

 %
 %

	   			    	   	        \draw (2) -- (c1);
	   	\draw (2) -- (c2);
	   		\draw (3) -- (c1);

	    %
	    %
	    %

	   	   	    	\end{tikzpicture}
	   	   	    	\end{center}
		
	   	   		\caption{The instance $I$ described in \Cref{ex:running}. The reservation graph $B_I$ on the left has an edge from $i$ to $c$ if $i$ is eligible for $c$. The priority rankings of the categories are depicted on the right.
}
	   			\label{fig:firstexample-0}
	   	   	\end{figure}

%
%

We are interested in allocation rules, which, for each instance, return a matching.

\begin{definition}[Allocation rule]
	An allocation rule maps every instance $I$ to a matching for $I$. 
\end{definition}

We say that an allocation rule $f$ satisfies one of the axioms in Definitions~\ref{def:eligibility} to~\ref{def:maximum} if $f(I)$ satisfies the axiom for all instances $I$.
Moreover, we define a notion of strategyproofness for allocation rules. 
Note that all units are identical and agents have no preferences over the category of the unit they receive.
However, they may have an incentive to hide being eligible for a category, or, more generally, to aim for a lower priority for some category.\footnote{In the context of school choice, lowering oneself in the priority ranking of a school is akin to students deliberately under-performing in an entrance exam.}

Formalizing strategyproofness requires the following definition.
Let $\prios$ and $\prios[']$ be priority profiles and $i \in N$.
We say agent $i$'s priority decreases from $\prios$ to $\prios[']$ if for all $j,k \neq i$ and $c\in C$,
\begin{align*}
	j \pref_c k &\longleftrightarrow j \pref_c' k\\
	j \pref_c i &\longrightarrow j \pref_c' i \text{ and } j \succ_c i \longrightarrow j \succ_c' i
\end{align*}
That is, the priority rankings over agents other than $i$ are the same in both profiles and $i$ can only move down in the priority rankings from $\prios$ to $\prios[']$.
We also say that $i$'s priority decreases from $I = (N,C,\prios,(q_c))$ to $I' = (N,C,\prios['],(q_c))$.
Strategyproofness requires that if $i$ is unmatched for $I$, then $i$ is also unmatched for $I'$.

\begin{definition}[Strategyproofness]
	An allocation rule $f$ is strategyproof if $f(I)(i) = \emptyset$ implies $f(I')(i) = \emptyset$ whenever $i$'s priority decreases from $I$ to $I'$.
\end{definition}	
In particular, with a strategyproof allocation rule, agents cannot benefit from hiding that they are eligible for a category.\footnote{{This restricted version of strategyproofness under which agents do not have an incentive to hide their eligible categories, has been referred to as incentive-compatibility by \citet{AyBo20a}.}}

An allocation rule is non-bossy if no unmatched agent can decrease her priority to change the set of matched agents. 
Combined with the other axioms, this property turns out to be fairly demanding. 
Thus, we weaken it by requiring only that no unmatched agent can decrease her priority and thereby change which of the agents lower in the baseline ordering are matched.
\begin{definition}[Weak non-bossiness]
	An allocation rule $f$ is weakly non-bossy if $f(I)(i) = \emptyset$ implies that any $j$ with $i\succ_\pi j$ is matched in $f(I)$ if and only if $j$ is matched in $f(I')$ whenever $i$'s priority decreases from $I$ to $I'$.
\end{definition}

\section{Issues with Breaking Ties and Applying the Deferred Acceptance Algorithm}\label{sec:da}

\citet{PSU+20a} showed that if one artificially introduces (strict) preferences for the agents over the categories they are eligible for and applies the Deferred Acceptance algorithm to the resulting two-sided matching problem, the resulting matching complies with the eligibility requirements, respects priorities, and is non-wasteful \citep[][Theorem 2]{PSU+20a}. 
They state the algorithmic implications of their result. 
\begin{quoting}
	``Not only is this result a second characterization of matchings that satisfy our three basic axioms, it also provides a concrete procedure to calculate all such matchings.''
	\end{quoting} 
	
Although considering all possible artificial preferences and running the Deferred Acceptance algorithm gives all the matchings satisfying the three axioms, it is computationally expensive to consider ${|C|!}^n$ different preference profiles. 
Moreover, not every preference profile leads to a compelling outcome even if the categories have strict priorities. 
For example, some preference profiles---like the one in the following example---lead to matchings that are not of maximum size.\footnote{The issue is also evident from the discussion by \citet{PSU+20a} where they point out that sequential treatment of categories may not give a maximum size matching.}


\begin{example}
%
Consider the instance in Example~\ref{ex:running}.
	Suppose we run the Deferred Acceptance algorithm assuming all agents prefer $c_1$ to $c_2$ to $c_3$. 
	Assuming agents can only be matched to categories they are eligible for (compliance with eligibility requirements), the resulting matching is $\mu_2 = \{\{2,c_1\}\}$.
	This matching is however not the most efficient use of the resources because it is possible to allocate all units while still satisfying the axioms in Definitions~\ref{def:eligibility} to~\ref{def:nonwasteful} by choosing $\mu_5 = \{\{3,c_1\},\{2,c_2\}\}$.
	\end{example}
	
Hence, artificially introducing preferences of agents and running the Deferred Acceptance algorithm can lead to inefficient allocations.
Even if we ignore computational concerns and can assign preferences to agents so that the matching selected by the Deferred Acceptance algorithm is of maximum size and respects priorities, it is not clear whether such a rule is strategyproof in the above sense. 
We propose a rule that circumvents both issues.

\section{The Reverse Rejecting Rule}

The Reverse Rejecting Rule (\rr) considers the agents in order of the baseline ordering $\succ_\pi$ from lowest to highest priority.
It rejects an agent if the agents who have not been rejected thus far can form a maximum size matching (among matchings in the reservation graph) for which none of the rejected agents has justified envy towards a matched agent.
The second constraint is implemented by removing those edges from the reservation graph $B_I$ that may cause justified envy by rejected agents.
That is, if agent $i$ is rejected, we cannot allow any agent $j$ to be matched to a category $c$ if $i\succ_c j$, and so remove the edge $\{j,c\}$ from $B_I$. 
More generally, if $R$ is the set of rejected agents, $B_I^{-R}=((N\setminus R)\cup C,E)$ with 
\[E=\{\{j,c\}\midd j\succ_c \emptyset \text{ and there is no } i\in R \text{ such that } i\succ_c j\}\]
is the corresponding reduced reservation graph.
Note that if $R\supset R'$, then 
$B_I^{-R}$ is a subgraph of $B_I^{-R'}$.
In particular, $B_I^{-R}$ is a subgraph of $B_I$.
The \rr rule then works as follows:
\begin{quoting}
	Let the set of rejected agents $R$ be empty at the start and consider the agents in order of the baseline ordering $\succ_\pi$ from lowest to highest priority.
	When considering agent $i$, add $i$ to the rejected agents $R$ if and only if $\ms(B_I^{-(R\cup\{i\})}) = \ms(B_I)$.
	After the last agent has been considered, let $R_I$ be the final set of rejected agents and choose a maximum size matching of the reduced reservation graph $B_I^{-R_I}$.
\end{quoting}

The methodology of the \rr rule is different from the horizontal envelop rule of \citet{SoYe20a} and the Smart Reserves rule of \citet{PSU+20a}. In contrast to iteratively or instantly selecting agents that will be matched, the \rr rule goes in the reverse order of the baseline ordering and decides which agents to reject. More importantly, the \rr rule works for heterogeneous priorities.

				   \begin{example}[Illustration of the \rr rule] \label{eg2:rr}
			   	Consider an instance with
			   	\[
			   	N = \{1,2,3,4\}, \quad C = \{c_1,c_2\},\quad q_{c_1} = q_{c_2} = 1
			   	\]
			   	The priorities are $1 \succ_{c_1} 4  \succ_{c_1} 2 \succ_{c_1} \emptyset$ and $1 \succ_{c_2} 3 \succ_{c_2} \emptyset$.
			   	For $1 \succ_\pi 2 \succ_\pi 3 \succ_\pi 4$, let us simulate the \rr rule. The reservation graph $B_I$ is depicted in \Cref{fig:eg2-1}.
					It has a maximum size matching of size 2. 
					First agent 4 is considered.
					Since the graph $B_I^{-\{4\}}$ depicted in \Cref{fig:eg2-2} admits a matching of size 2, agent 4 is placed in $R$. 
					Next, agent 3 is considered and not placed in $R$ since the graph $B_I^{-\{3,4\}}$ depicted in \Cref{fig:eg2-3} does not admit a matching of size 2.
					The graph $B_I^{-\{2,4\}}$ depicted in \Cref{fig:eg2-4} admits a matching of size 2.
					Hence, agent 2 is placed in $R$. 
					Lastly, since $B_I^{-\{1,2,4\}}$ depicted in \Cref{fig:eg2-5} does not admit a matching of size 2, agent 1 is not placed in $R$.
	The final outcome is a maximum size matching of the graph $B_I^{-\{2,4\}}$ as depicted in \Cref{fig:eg2-4}.
\begin{figure}[]
	\centering
	\def\dist{1.2}
	\begin{subfigure}[t]{0.32\textwidth}
		\centering	   	    		\begin{tikzpicture}[-,>=stealth',shorten >=1pt,auto,node distance=3cm,
					   	   	    	        thick,main node/.style={circle,fill=white!20,draw,minimum size=0.5cm,inner sep=0pt]}, scale=0.9]

					   	   	    	    \node[main node] at (0, 0) (1)  {$1$};
					   	   	    	    \node[main node] at (0, -\dist) (2)  {$2$};
					   	   	    	    \node[main node] at (0, -2*\dist) (3)  {$3$};
										    \node[main node] at (0, -3*\dist) (4)  {$4$};

					   	   	    	    \node[main node] at (3.5*\dist, -\dist) (c1)  {$c_1$};
					   	   	    	    \node[main node] at (3.5*\dist, -2*\dist) (c2)  {$c_2$};
						
												\node[] at (3.5*\dist, -2.6*\dist) {{$q_{{c_2}=1}$}};
						
		
						\node[] at (3.5*\dist, -0.4*\dist) {{$q_{{c_1}=1}$}};
		

\draw[dotted] (1) -- (c1);					   	 \draw[dotted] (2) -- (c1);
\draw[dotted] (4) -- (c1);
	\draw[dotted] (1) -- (c2);
	\draw[dotted] (3) -- (c2);

					   	   	    	\end{tikzpicture}
			\caption{$B_I$
			}
				   			\label{fig:eg2-1}
			\end{subfigure}
			\begin{subfigure}[t]{0.32\textwidth}	
		\centering							   	   	    	\begin{tikzpicture}[-,>=stealth',shorten >=1pt,auto,node distance=3cm,
									   	   	    	        thick,main node/.style={circle,fill=white!20,draw,minimum size=0.5cm,inner sep=0pt]}, scale=0.9]

									   	   	    	    \node[main node] at (0, 0) (1)  {$1$};
									   	   	    	    \node[main node] at (0, -\dist) (2)  {$2$};
									   	   	    	    \node[main node] at (0, -2*\dist) (3)  {$3$};

									   	   	    	    \node[main node] at (3.5*\dist, -1*\dist) (c1)  {$c_1$};
									   	   	    	    \node[main node] at (3.5*\dist, -2*\dist) (c2)  {$c_2$};
						
																\node[] at (3.5*\dist, -2.6*\dist) {{$q_{{c_2}=1}$}};
						
		
										\node[] at (3.5*\dist, -0.4*\dist) {{$q_{{c_1}=1}$}};
		

				\draw[dotted] (1) -- (c1);					   	 
					\draw[dotted] (1) -- (c2);
					\draw[dotted] (3) -- (c2);

									   	   	    	\end{tikzpicture}
				   	   		\caption{$B_I^{-\{4\}}$ 
			}
				   			\label{fig:eg2-2}
			\end{subfigure}
			\begin{subfigure}[t]{0.32\textwidth}	
				\centering
				 \begin{tikzpicture}[-,>=stealth',shorten >=1pt,auto,node distance=3cm,
									   	   	    	        thick,main node/.style={circle,fill=white!20,draw,minimum size=0.5cm,inner sep=0pt]}, scale=0.9]

									   	   	    	    \node[main node] at (0, 0) (1)  {$1$};
									   	   	    	    \node[main node] at (0, -\dist) (2)  {$2$};

									   	   	    	    \node[main node] at (3.5*\dist, -1*\dist) (c1)  {$c_1$};
									   	   	    	    \node[main node] at (3.5*\dist, -2*\dist) (c2)  {$c_2$};
						
																\node[] at (3.5*\dist, -2.6*\dist) {{$q_{{c_2}=1}$}};
						
		
										\node[] at (3.5*\dist, -0.4*\dist) {{$q_{{c_1}=1}$}};
		

				\draw[dotted] (1) -- (c1);					   	 
					\draw[dotted] (1) -- (c2);

									   	   	    	\end{tikzpicture}
									   	   		\caption{$B_I^{-\{3,4\}}$ 
								}
									   			\label{fig:eg2-3}
			\end{subfigure}
			\par\bigskip
			\begin{subfigure}[t]{0.32\textwidth}	
				\centering	   	   	    	\begin{tikzpicture}[-,>=stealth',shorten >=1pt,auto,node distance=3cm,
									   	   	    	        thick,main node/.style={circle,fill=white!20,draw,minimum size=0.5cm,inner sep=0pt]}, scale=0.9]

									   	   	    	    \node[main node] at (0, 0) (1)  {$1$};
									   	   	    	    \node[main node] at (0, -2*\dist) (3)  {$3$};

									   	   	    	    \node[main node] at (3.5*\dist, -1*\dist) (c1)  {$c_1$};
									   	   	    	    \node[main node] at (3.5*\dist, -2*\dist) (c2)  {$c_2$};
						
																\node[] at (3.5*\dist, -2.6*\dist) {{$q_{{c_2}=1}$}};
						
		
										\node[] at (3.5*\dist, -0.4*\dist) {{$q_{{c_1}=1}$}};
		

				\draw[dotted] (1) -- (c1);					   	 
					\draw[dotted] (1) -- (c2);
					\draw[dotted] (3) -- (c2);

									   	   	    	\end{tikzpicture}
		
									   	   		\caption{$B_I^{-\{2,4\}}$ 
								}
									   			\label{fig:eg2-4}
			\end{subfigure}
			\begin{subfigure}[t]{0.32\textwidth}	
				\centering
									   	   	    	\begin{tikzpicture}[-,>=stealth',shorten >=1pt,auto,node distance=3cm,
									   	   	    	        thick,main node/.style={circle,fill=white!20,draw,minimum size=0.5cm,inner sep=0pt]}, scale=0.9]

									   	   	    	    \node[main node] at (0, -2*\dist) (3)  {$3$};

									   	   	    	    \node[main node] at (3.5*\dist, -1*\dist) (c1)  {$c_1$};
									   	   	    	    \node[main node] at (3.5*\dist, -2*\dist) (c2)  {$c_2$};
					
																\node[] at (3.5*\dist, -2.6*\dist) {{$q_{{c_2}=1}$}};
					
	
										\node[] at (3.5*\dist, -0.4*\dist) {{$q_{{c_1}=1}$}};
	

					\draw[dotted] (3) -- (c2);

									   	   	    	\end{tikzpicture}
	
									   	   		\caption{$B_I^{-\{1,2,4\}}$
								}
									   			\label{fig:eg2-5}
			\end{subfigure}
			\caption{Graphs for the instance $I$ in \Cref{eg2:rr}.}
		\end{figure}
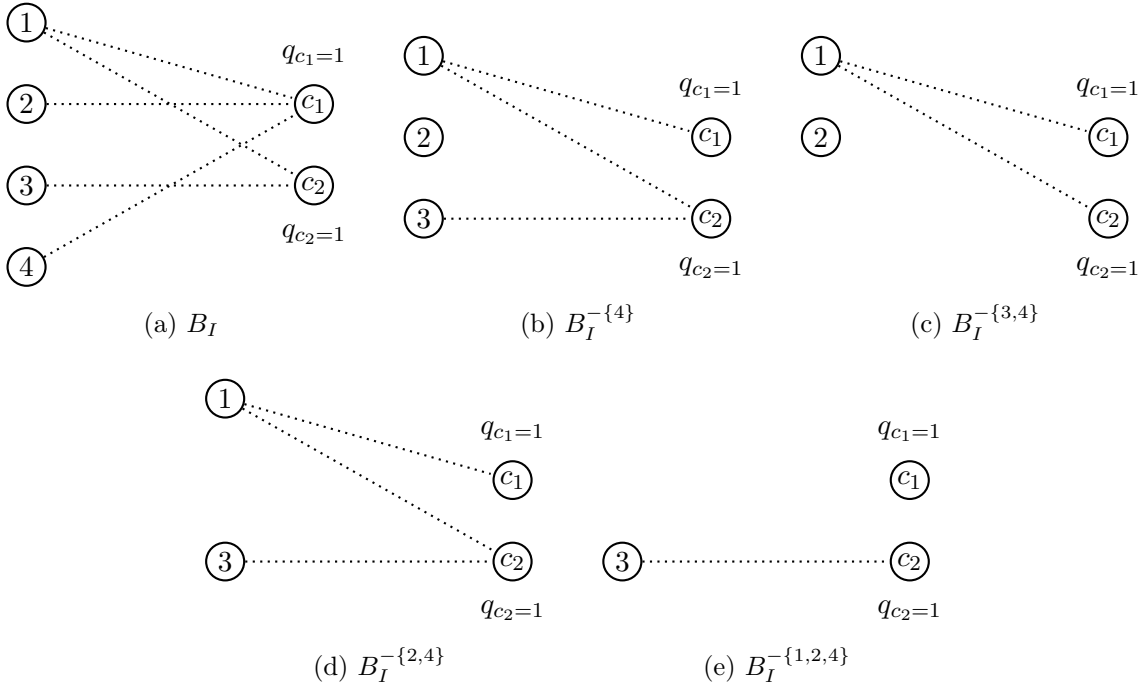
			     \end{example}


		\begin{theorem}[Properties of the \rr rule]\label{thm:allaxioms-rr}
			The \rr rule
			\begin{enumerate}[label=(\roman*)]
				\item complies with eligibility requirements, 
				\item respects priorities,  
				\item returns a matching of maximum size among feasible matchings, 
				\item is weakly non-bossy,
				\item is strategyproof, and
				\item can be computed in strongly polynomial time. 
			\end{enumerate}
		\end{theorem}
		\begin{proof}
			\begin{enumerate}[label=(\roman*)]
				
			\item Compliance with eligibility requirements. 
			The outcome is a matching of $B_I^{-R_I}$, which is a subgraph of $B_I$.
			Any edge $\{i,c\}$ of $B_I$ is such that $i\succ_c \emptyset$. 
			Therefore the outcome satisfies eligibility requirements.
				
				\item Respect of priorities. 
				Let $\mu$ be the matching returned by the \rr rule for the instance $I$.
				First, we claim that $\mu$ provides no justified envy for any agent in $R_I$. 
				If $j\in R_I$, $i\in N$, and $c\in C$ with $\mu(i) = c$ and $j\succ_c i$, then $\{i,c\}$ was deleted from the reservation graph when rejecting $j$ and is thus not an edge of $B_I^{-R_I}$.
				Since $\mu$ is a matching of $B_I^{-R}$, it does not match $i$ to $c$.

				Second, we show that $\mu$ matches all agents in $N \setminus R_I$.
				In particular, no agent in $N\setminus R_I$ can have (justified) envy.
				Assume for contradiction that $|\mu| < n - |R_I|$.
				This means that $\ms(B_I^{-(R_I\cup\{i\})}) < \ms(B_I) = \ms(B_I^{-R_I}) = |\mu|$ for all $i\in N\setminus R_I$.
				Among all maximum size matchings of $B_I^{-R_I}$, let $\mu'$ be one that is Pareto optimal with respect to the priorities.
				If there were $i,j\in N\setminus R_I$ so that $j$ has justified envy towards $i$ in $\mu'$ ($\mu'(i) = c$, $\mu'(j) = \emptyset$, and $j\succ_c i$), then the matching $\mu' \setminus \{\{i,c\}\}\cup\{\{j,c\}\}$ would Pareto dominate $\mu'$ with respect to the priorities, which cannot be.
				Thus, $\mu'$ respect priorities, so that for $j\in N \setminus R_I$ with $\mu'(j) = \emptyset$, $\mu'$ is a matching of $B_I^{-(R_I\cup\{j\})}$.
				It follows that $\ms(B_I^{-(R_I\cup\{j\})}) = \ms(B_I^{-R_I})$, which contradicts that the \rr rule rejects none of the agents in $N\setminus R_I$.
			
				
				\item Maximum size matching. The returned matching is by construction a matching of size $\ms(B_I)$ and hence of maximum size among all matching complying with the eligibility requirements.

				\item Weak non-bossiness. 
				For simplicity, suppose the baseline ordering $\succ_\pi$ is $1\succ_\pi 2\succ_\pi \dots \succ_\pi n$.
				Let $i\in N$, $I,I'$ be instances so that $i$'s priority decreases from $I$ to $I'$.
				Assume that $i$ is unmatched in $\rr(I)$.
				Note that $\ms(B_I) = \ms(B_{I'})$ since the edge set of $B_{I'}$ is a subset of the edge set of $B_I$ and $\rr(I)$ does not match $i$.
				Let $R_I^i = \{j\in R_I\colon i\succ_\pi j\}$ be the agents who are unmatched in $\rr(I)$ and are lower in the baseline ordering than $i$. 
				Define $R_{I'}^i$ similarly.
				We need to show that $R_I^i = R_{I'}^i$.
				To this end, we prove by induction that $j\in R_{I}^i$ if and only if  $j\in R_{I'}^i$ for $j=n$ to $i+1$ 
				
				The base case is trivial since the set of rejected agents is the empty set at the start of the algorithm under both instances $I$ and $I'$. 
		
				For the induction step, let $j > i$ and suppose $j'\in R_{I}^i$ if and only if  $j'\in R_{I'}^i$ for $j' = n$ to $j+1$. 
							We prove that $j\in R_{I}^i$ if and only if  $j\in  R_{I'}^i$. 
							Denote by $R$ the set of rejected agents at the beginning of the round where agent $j$ is considered. 
							Note that $R=R_{I}^i\cap \{j+1,\dots,n\} = R_I \cap \{j+1,\dots,n\}$.
				   		If $j\in R_I^i$, this is because there exists a matching $\mu$ of $B_I^{-(R\cup \{j\})}$ 
						with $|\mu| = \ms(B_I)$. 
						Since $R_I\cap \{j,\dots,n\} \subset R\cup \{j\}$, $\mu$ does not match any agent in $R_I\cap \{j,\dots,n\}$.
						Since $i$ is unmatched in $\rr(I)$, it follows that $i \in R_I$, and we can assume that $\mu$ does not match $i$.
				Now consider $B_{I'}^{-(R\cup \{j\})}$. 
				For each $k\in N\setminus (R\cup \{j\} \cup \{i\})$ and $c\in C$, $\{k,c\}$ is an edge of $B_{I}^{-(R\cup \{j\})}$ if and only if $\{k,c\}$ is an edge of $B_{I'}^{-(R\cup \{j\})}$ since the priorities for all agents other than $i$ are the same in $I$ and $I'$.
				It follows that $\mu$ is also a feasible matching of $B_{I'}^{-(R\cup \{j\})}$. 
				Therefore $j\in R_{I'}^i$.

				   		For the other direction, suppose $j\in R_{I'}^i$. 
							Hence, there exists a matching $\mu$ of $B_{I'}^{-(R\cup \{j\})}$ with $|\mu| =  \ms(B_{I'}) = \ms(B_I)$. 
							As above, for each $k\in N\setminus (R\cup \{j\} \cup \{i\})$ and $c\in C$, $\{k,c\}$ is an edge of $B_{I}^{-(R\cup \{j\})}$ if and only if $\{k,c\}$ is an edge of $B_{I'}^{-(R\cup \{j\})}$.
				For $i$ and any $c\in C$, note that if $\{i,c\}$ is an edge of $B_{I'}^{-(R\cap \{j\})}$, 
				then $\{i,c\}$ is an edge of $B_{I}^{-(R\cup \{j\})}$. 
				Hence $\mu$ is also a feasible matching of $B_{I}^{-(R\cup \{j\})}$. 
				Thus, $j\in R_{I'}^i$.
				
				This completes the proof that $R_{I}^i=R_{I'}^i$.
				We remark that the argument above breaks down for agents $j$ with $j \succ_\pi i$ since if $i\in S$, not all edges of $B_{I'}^{-S}$ are also edges of $B_I^{-S}$.

				\item Strategyproofness.  
				Suppose agent $i$'s priority decreases from $I$ to $I'$ and $i$ is unmatched in $\rr(I)$.
				As shown in the proof for respect of priorities, this is equivalent to $i\in R_I$.
				We prove that $i\in R_{I'}$. 
				By the proof of weak non-bossiness above, $R_I^{i}=R_{I'}^i$.
				Since $i\in R_I$, it follows that $\ms(B_I^{-(R_I^i\cup \{i\})})=\ms(B_I)$.
				Note that $B_{I}^{-(R_I^i\cup \{i\})}$ is a subgraph of $B_{I'}^{-(R_{I'}^i\cup \{i\})}$ since $R_I^{i}=R_{I'}^i$ and $i$'s priority decreases from $I$ to $I'$.
				It follows that $\ms(B_{I'}^{-(R_{I'}^i\cup \{i\})}) = \ms(B_{I}^{-(R_I^i\cup \{i\})}) = \ms(B_I)$. 
				Therefore $i\in R_{I'}$.

\item Polynomial time computability. 
The rule makes at most $n$ calls to computing a maximum size matching of the underlying reservation graphs thus runs in strongly polynomial time. 

			\end{enumerate}
			\end{proof}	
			
			\begin{remark}[The \rr rule is not non-bossy] 
				Consider an instance with
				\[
				N = \{1,2,3,4\}, \quad C = \{c_1,c_2\},\quad q_{c_1} = q_{c_2} = 1
				\]
				The priorities are $1 \succ_{c_1} 4  \succ_{c_1} 2 \succ_{c_1} \emptyset$ and $1 \succ_{c_2} 3 \succ_{c_2} \emptyset$.
				For $1 \succ_\pi 2 \succ_\pi 3 \succ_\pi 4$, the \rr rule yields the matching $\mu = \{\{1,c_1\},\{3,c_2\}\}$. 
				If agent $4$ reports that they are not eligible for $c_1$, agent 2 moves to the second equivalence class in the priority order of $c_1$ and the \rr rule yields the matching $\mu' = \{\{1,c_2\},\{2,c_1\}\}$.
				Since $\mu\neq\mu'$, the \rr rule violates non-bossiness. 
			\end{remark}
			
\begin{remark}[Relation to school choice]
	\Cref{thm:allaxioms-rr} can be rephrased in the context of school choice~\citep[see, e.g.,][]{AbSo03b}.
	Agents correspond to students and categories correspond to schools.
	Each student finds a subset of schools acceptable and is indifferent between all acceptable schools.
	The schools have priorities over the students. 
	Then, \Cref{thm:allaxioms-rr} reads as follows. 
	\begin{theorem*} 
		Consider the school choice problem where the students partition the schools into acceptable and unacceptable schools.
		Then, there is an allocation rule that only matches students to acceptable schools, admits no justified envy, is non-wasteful, matches the maximal feasible number of students, and is strategyproof for students. 		
				\end{theorem*}	
\end{remark}
			
			Note that among the matchings that comply with the eligibility requirements, respect priorities, and have maximum size among feasible matchings, the \rr rule returns one that matches the set of agents that is maximal according to the upward lexicographic ordering on subsets of agents induced by the baseline ordering. 
	Hence, the outcome of the \rr rule depends heavily on the baseline ordering.  
We show that in fact, \emph{every} matching satisfying these three properties is an outcome of the \rr rule for some baseline ordering.
Thus, these properties characterize all possible outcomes of the \rr rule.
			   
		\begin{theorem}[Characterization of \rr outcomes]\label{thm:rr-charac}
			A matching complies with the eligibility requirements, respects priorities, and has maximum size among feasible matchings if and only if it is a possible outcome of the \rr rule for some baseline ordering.
			\end{theorem}
			\begin{proof}
				Consider a matching $\mu$ with the three properties. Suppose it matches the set of agents $S\subseteq N$. Our first claim is that $\mu$ is feasible matching of $B_{I}^{-(N\setminus S)}$. Since $\mu$ satisfies the eligibility requirements, it is a matching of the graph $B_I$ constrained to the vertex set $S\cup C$.
				Since $\mu$ respects priorities, there exists no edge $\{i,c\}\in \mu$ such that $j\succ_c i$ for some $j\in N\setminus S$. Therefore, it follows that $\mu$ is a matching of $B_{I}^{-(N\setminus S)}$. 
				
				Now consider a baseline ordering $\succ_\pi$ such that $i\succ_\pi j$ for all $i\in S$ and $j\in N\setminus S$.  
				We claim that $R_I = N\setminus S$.
				Since $|\mu| = \ms(B_I)$ by assumption and $\mu$ is a matching of $B_I^{-(N\setminus S)}$ as shown above, it follows that $N\setminus S\subset R_I$.
				The inclusion cannot be strict since then $\ms(B_I^{-R_I}) < |\mu| = \ms(B_I)$.
				Thus, $R_I = N\setminus S$.
				Since $\mu$ is a maximum size matching of $B_I^{-(N\setminus S)}$, it is a possible outcome of the \rr rule.
				
				
				The converse follows from \Cref{thm:allaxioms-rr}.
				\end{proof}


	\section{Treating Reserved and Unreserved Units Asymmetrically}

	We have thus far treated all categories symmetrically.
	Now we designate one category $c_u\in C$ as the \emph{unreserved category}.
	All agents are eligible for the unreserved category and the priority ranking for the unreserved category equals the baseline ordering $\succ_\pi$.
	We refer to the units reserved for the unreserved category as \emph{unreserved units} and call the remaining categories \emph{preferential categories}. {The set of preferential categories is denoted by $C_p$.}
	There are two reasons for introducing the unreserved category. 
	Firstly, we want to consider maximum beneficiary assignments, which maximize the number of agents matched to preferential categories.\footnote{{In our context, the combination of maximum beneficiary assignment and non-wastefulness implies maximum size.}}
	\begin{definition}[Maximum beneficiary assignment]\label{def:maximumb}
	A matching $\mu$ is a \emph{maximum beneficiary assignment} if it maximizes the number of agents matched to a preferential category.
	\end{definition}
	

	Secondly, in various reserves problems, unreserved units are treated in special ways such as being allocated later or earlier. 
	We discuss both of these issues.

	
We observe that applying the \rr rule to the preferential categories $C_p$ and then allocating the unreserved units among the unmatched agents, say, according to the baseline ordering, yields a maximum beneficiary assignment.



	%
	


\subsection{Order Preservation}
	Certain policy goals may require allocating a designated number of unreserved units \emph{before} allocating the units reserved for preferential categories. For example, the rationale for the ``over-and-above'' reserve approach is that agents first get a bite at the designated unreserved units before they utilize the preferential category units for which they are eligible.
	By contrast, the ``minimum-guarantees'' reserve approach first assigns agents to preferential categories and then matches the remaining agents to the unreserved units.
	We first define minimum-guarantees and over-and-above reservation rules~\citep[Chapter 13, Part B]{Gala84a} when the agents are eligible for at most one preferential category and all categories have priorities that are consistent with the baseline ordering in the sense that the agents that are eligible for a category are ranked according to the baseline ordering.\footnote{\citet{Gala61a,Gala84a} was one of the first to study the differences between the minimum-guarantees and over-and-above reservation rules in depth.}

	\paragraph{Minimum-guarantee} Consider the agents in the order of the baseline ordering. 
	For each agent, match her to a preferential category if $(i)$ the agent is eligible for a preferential category and $(ii)$ not all units reserved for this category have been allocated. 
	Otherwise match the agent to the unreserved category if an unreserved unit remains.\footnote{{There is another version of the minimum-guarantees rule called the \emph{Partha method} that gives an equivalent outcome but operates differently as an algorithm. In the Partha method, the units are allocated according to the baseline ordering (``merit'') and preferential reservation is only enforced if the reserves are not maximally used~\citep{Gala84a}.}}


	\paragraph{Over-and-Above} 
	Consider the agents in the order of the baseline ordering. 
	For each agent $i$, match her to the unreserved category only if $(i)$ an unreserved unit remains and $(ii)$ if agent $i$ is eligible for some preferential category, say, $c$, then there still are at least $q_c$ agents from $N_c$ other than $i$ who are unmatched.
	Then, fill up the preferential categories as follows: for each $c\in C_p$, the $\min\{q_c, |N_c|\}$ highest priority unmatched agents are given one unit each from $c$.

	\medskip

	{We present an example adapted from the book of \citet{Gala84a} to illustrate the difference between the minimum-guarantees and the over-and-above approach.\footnote{\citet{Gala84a} studied minimum-guarantees and over-and-above in the context of job allocation in India where the baseline ordering represents the merit ranking and the preferential categories are historically disadvantaged groups. He observes that the minimum-guarantees rule ``insures that the amount of effective reservation is somehow commensurate with the backwardness that inspired it.'' On the other hand, he observes that the minimum-guarantees rule may ``overstate the effective amount of reservation'' especially if the disadvantaged groups are doing well enough on merit (page 461).}}
	\begin{example}
	Consider the case where $N=\{1,2,3,4\}$, $C=\{c,c_u\}$, $q_c=q_{c_u}=1$, $N_c=\{1,4\}$, and $1\succ_\pi2\succ_\pi3\succ_\pi4$.
	The outcome of the minimum-guarantees rule is that agent $1$ and $2$ are selected and the matching is $\{\{1,c\},\{2,c_u\}\}.$
 The outcome of the over-and-above rule is that agents $1$ and $4$ are selected and the matching is $\{\{1,c_u\},\{4,c\}\}$.
 In the example, the minimum-guarantees rule coincides with the rule that solely uses the baseline ordering. On the other hand, the over-and-above rule provides additional representation of agents from the preferential category~$c$.
	\end{example}
	
	We note that the minimum-guarantees approach allocates the unreserved units at the end whereas the over-and-above approach allocates the unreserved units first.
We now explicitly distinguish between unreserved units that are processed earlier and later. To be precise, let $C = C_p \cup\{c_u^1,c_u^2\}$, where $c_u^1$ represents the unreserved units to be treated first and $c_u^2$ the unreserved units to be treated at the end. 
	We view $c_u^1$ and $c_u^2$ as subcategories of the unreserved category $c_u$ and the $q_{c_u}$ slots for $c_u$ as being partitioned into $q_{c_u^1}$ slots for $c_u^1$ and $q_{c_u^2}$ slots for $c_u^2$.
	Hence, $q_{c_u^1}+q_{c_u^2}=q_{c_u}$ and ${\succ_{c_u^1}}={\succ_{c_u^1}}={\succ_\pi}$. 
	If an agent is matched to $c_u^1$ or $c_u^2$, we say that she receives an unreserved unit. 

	
	\citet{PSU+20a} formulated a family of rules, called \emph{Smart Reserves rules}, that allow agents to be eligible for multiple preferential categories and generalize the minimum-guarantee and over-and-above rules to this case.
	In this section, we capture these approaches via an axiom for matchings called \emph{order preservation} and then propose a new rule that also works for heterogeneous priorities. Order preservation is parametrized by the number of unreserved units that are placed in category $c_u^1$ and $c_u^2$. {It captures the idea that there is an order of the categories ($c_u^1$, $C_p$, and then $c_u^2$) and no two agents should be able to swap their matches so that eligibility requirements are not violated, and an earlier category gets a higher priority agent after the swap.}

	\begin{definition}[Order preservation]\label{def:orderpreserve}
		Consider a matching $\mu$ of agents to categories in $C_p\cup \{c_u^1, c_u^2\}$. We say that $\mu$ is \emph{order preserving} (with respect to $c_u^1$ and $c_u^2$) for baseline ordering $\succ_\pi$ if 
		for any two agents $i,j\in N$,
		\begin{enumerate}[label=(\roman*)]
			\item $\mu(i) \in C_p\cup\{c_u^2\}$, $i\succ_{\mu(j)}j$, and $j$ is eligible for category $\mu(i)$ implies $\mu(j) \neq c_u^1$, and \label{item:preservation1}
			\item $\mu(j)\in C_p\cup \{c_u^1\}$, $i\succ_{\mu(j)}j$, and $i$ is eligible for category $\mu(j)$ implies $\mu(i)\neq c_u^2$.\footnote{It follows from $i\succ_{\mu(j)} j$ that $i$ is eligible for category $\mu(j)$. We state it explicitly in \ref{item:preservation2} to maintain the symmetry with \ref{item:preservation1}.} 
			\label{item:preservation2}
		\end{enumerate}
		\end{definition}

	There are two extreme ways unreserved units can be treated under order preservation. The first is if $q_{c_u^1}=0$ and $q_{c_u^2}=q_{c_u}$. The other is if $q_{c_u^1}=q_{c_u}$ and $q_{c_u^2}=0$. 
	The conceptual contribution of Definition~\ref{def:orderpreserve} is that instead of describing over-and-above and minimum-guarantees rules as consequences of certain sequential allocation methods, order preservation captures a key property of their resulting \emph{matchings}. 
	It is formulated so that it allows for heterogeneous priorities or agents being eligible for multiple categories. 
	
	We collect some of the properties of the minimum-guarantees rule and the over-and-above rule when each agent is eligible for at most one preferential category and categories have consistent priorities.
	These properties characterize the minimum-guarantees rule for the corresponding notion of order-preservation.

	\begin{proposition}[Properties of minimum-guarantees and over-and-above]
		Assume each agent is eligible for \emph{at most one} preferential category and all categories have \emph{consistent priorities}. 
		Then, the outcome of the minimum-guarantees (over-and-above) rule
		\begin{enumerate}[label=(\roman*)]
			\item complies with the eligibility requirements,\label{item:eligible}
			\item is a maximum beneficiary assignment,
			\item respects priorities, 
			\item is non-wasteful, and
			\item satisfies order preservation for $q_{c_u^1}=0$ and $q_{c_u^2}=q_{c_u}$ ($q_{c_u^1}=q_{c_u}$ and $q_{c_u^2}=0$).\label{item:orderpreservation}
		\end{enumerate}
	Moreover, the outcome of the minimum-guarantees rule is the unique matching satisfying~\ref{item:eligible} to~\ref{item:orderpreservation} with $q_{c_u^1}=0$ and $q_{c_u^2}=q_{c_u}$.
		\end{proposition}
		\begin{proof}
			First consider the minimum-guarantees rule.
			It complies with the eligibility requirements and is non-wasteful. 
			It also yields a maximum beneficiary assignment because for each preferential category, the maximum possible number of agents is matched. 
			The unreserved units are matched later to the unmatched agents. 
			Therefore, the matching respects priorities and satisfies order preservation for $q_{c_u^1}=0$ and $q_{c_u^2}=q_{c_u}$. 
			
			We prove that there is exactly one matching satisfying~\ref{item:eligible} to~\ref{item:orderpreservation} 			with $q_{c_u^1}=0$ and $q_{c_u^2}=q_{c_u}$.
			Hence, the outcome of the minimum-guarantees rule is the unique such matching.
Suppose for contradiction there are two such matchings $\mu'$ and $\mu''$. 
			Since each agent is eligible for at most one category in $C_p$ and $\mu'$ and $\mu''$ satisfy maximum beneficiary assignment, it follows that $|\mu'(c)|=|\mu''(c)|$ for all $c \in C_p$. 
			We prove that for either of $\mu'$ and $\mu''$, the agents matched to $c\in C_p$ are the $\min\{q_c,|N_c|\}$ highest priority eligible agents for $c$. 
			Suppose this is not the case for $\mu'$. 
			Then, there exist $i,j\in N_c$ such that $\mu'(j) = c$, $\mu'(i) \neq c$, and $i\succ_c j$. 
			If $\mu'(i) = \emptyset$, then $\mu'$ does not respect priorities. 
			If  $\mu'(i) = c_u^2$, then $\mu'$ does not satisfy order preservation as $i$ and $j$ can swap their matches without violating eligibility requirements. 
			We have established that for each $c\in C_p$, $\mu'(c)=\mu''(c)$. 
			Respect of priorities, non-wastefulness, and the fact that every agent is eligible for $c_u^2$ imply that the agents matched to $c_u^2$ for either of $\mu'$ and $\mu''$ are the highest priority agents who are not matched to categories in $C_p$.
			Hence, $\mu'(c_u^2)=\mu''(c_u^2)$.
			
			Now consider the over-and-above rule.
			It complies with the eligibility requirements and is non-wasteful.
				It also yields a maximum beneficiary assignment because for each preferential category, the maximum possible number of agents are matched.
				The unreserved units are matched to the highest priority agents possible subject to enabling a maximum beneficiary assignment.
				Therefore, the matching respects priorities and satisfies order preservation for $q_{c_u^1}=q_{c_u}$ and $q_{c_u^2}=0$.
 \end{proof}

 \begin{remark}[Non-uniqueness of over-and-above]
	 The outcome of the over-and-above rule is not necessarily the unique matching satisfying~\ref{item:eligible} to~\ref{item:orderpreservation} with $q_{c_u^1}=q_{c_u}$ and $q_{c_u^2}=0$.
	Let $N=\{1,2,3,4\}$ and $C=\{c_u^1, c_1, c_2\}$ with each category having capacity 1. Suppose the priorities are  
	$1 \succ_{c_u^1} 2 \succ_{c_u^1} 3 \succ_{c_u^1} 4 \succ_{c_u^1} \emptyset $, $1 \succ_{c_1} 3 \succ_{c_1} \emptyset$, and $2 \succ_{c_2} 4 \succ_{c_2} \emptyset$. 
	Then, the outcome of the over and above rule is $\{\{1,c_u^1\}, \{3,c_1\}, \{2,c_2\}\}$. 
	Another matching that satisfies the properties is $\{\{2,c_u^1\}, \{1,c_1\}, \{4,c_2\}\}$.
 \end{remark}

	\subsection{Smart Reverse Rejecting Rule}

 	The Smart Reserves rule of \citet{PSU+20a} gives agents the  unreserved units from $c_u^1$ as long as the set of remaining agents can be matched to get a maximum beneficiary assignment. 
	Whereas the \rr rule is not equipped to handle unreserved categories, the Smart Reserves rule is not designed to handle heterogeneous priorities. 
	The ideas of both approaches can be combined to obtain a \emph{Smart Reverse Rejecting (\srr)} rule.
	It first determines which agents get an unreserved unit from $c_u^1$, then allocates the units reserved for preferential categories to the remaining agents using the \rr rule, and gives the $c_u^2$ units to the remaining agents according to the baseline ordering.
	More precisely, the \srr rule works as follows:
\begin{quoting}
	Let the set of agents to be given unreserved units from $c_u^1$, $N_1$, be empty at the start and consider the agents in order of the baseline ordering $\succ_\pi$ from highest to lowest priority.
	When agent $i$ is considered, add $i$ to $N_1$ if $N_1$ contains fewer than $q_{c_u^1}$ agents and the agents the agents in $N\setminus (N_1\cup\{i\})$ can form a maximum beneficiary assignment.
	After the last agent has been considered, give each agent in $N_1$ an unreserved unit from $c_u^1$.
	Use the \rr rule to allocate the units reserved for the preferential categories $C_p$ to the remaining agents.
	Lastly, give the unreserved units from $c_u^2$ to the remaining agents in order of the baseline ordering from highest to lowest priority.
\end{quoting}

We show that the \srr rule preserves the properties of the \rr rule while allowing to treat reserved and unreserved units asymmetrically.
\Cref{tab:comparisonofrules} summarizes the properties of the \srr rule, the Smart Reserves rule, and the approach based on the Deferred Acceptance algorithm described in \Cref{sec:da}.

					   \begin{table}[t]
					   \centering
					   \begin{tabular}{lccc}
		\toprule

		&\srr&Smart Reserves&DA\\
		\midrule
		compliance with eligibility requirements&\checkmark&\checkmark&\checkmark\\
		maximum beneficiary assignment&\checkmark&\checkmark&--\\
		respect of priorities&\checkmark &\checkmark$^*$&\checkmark\\
		strategyproofness &\checkmark&n/a&\checkmark\\
			weak non-bossiness&\checkmark&n/a&\checkmark\\
			order preservation&\checkmark&\checkmark$^*$&--\\
			polynomial-time computability&\checkmark&\checkmark&\checkmark\\
					   	\bottomrule

					   \end{tabular}
					   \caption{Properties satisfied by the \srr rule, the Smart Reserves rule of \citet{PSU+20a}, and the Deferred Acceptance algorithm described in \Cref{sec:da}. An asterisk indicates that the property holds if priorities are strict and consistent with the baseline ordering. 
					   N/a indicates that the rule assumes homogenous priorities but the property allows for changes in the priorities that may result in inhomogeneous priorities.
					   }
					   \label{tab:comparisonofrules}
					   \end{table}


			  \begin{theorem}[Properties of the \srr rule]\label{thm:allaxioms:s-rev}
			  	The \srr rule
			  	\begin{enumerate}[label=(\roman*)]
			  		\item complies with eligibility requirements, \label{item:sreveligible} 
					\item yields a maximum beneficiary assignment,\label{item:srevmaxbene}
			  		\item respects priorities,  \label{item:srevprio}
			  		\item is strategyproof,\label{item:srevstrategyproof}
					\item is weakly non-bossy, \label{item:srevnonbossy}
					\item satisfies order preservation, and\label{item:srevorderpreservation}
					\item is polynomial-time computable.\label{item:srevpolytime}
			  	\end{enumerate}
			  \end{theorem} 
			  
\begin{proof}
	 \ref{item:sreveligible} and \ref{item:srevmaxbene} are clear by construction.
	 Note that a maximum beneficiary assignment is a maximum size matching of agent to all categories in $C = C_p\cup\{c_u^1,c_u^2\}$ since every agent is eligible for $c_u^1$ and $c_u^2$.
	 
	 \ref{item:srevprio} 	We first prove that no unmatched agent can have justified envy for an agent matched to $c_u^1$. Suppose an unmatched agent $i$ has higher baseline priority than an agent $j$ matched to $c_u^1$. 
	 Then, $i$ would have been considered before $j$ when determining which agents get units from $c_u^1$.
	  The chosen matching shows that neither $i$ nor the agents in $N_1$ (at the time $i$ was considered) are needed a form a maximum beneficiary assignment.
	  Hence, $i$ would have been added to $N_1$ and received a unit from $c_u^1$, which is a contradiction. 
	  		
	Second, no unmatched agent can have justified envy towards an agent matched to $c_u^2$ because each unmatched agent comes later in the baseline ordering than each agent matched to $c_u^2$. 
	Finally, no unmatched agent can have justified envy towards an agent matched to a category in $C_p$ as this would 
	contradict the fact that the \rr rule respects priorities. 
	 
	 \ref{item:srevstrategyproof} 						We show that if an agent $i$ is unmatched under the \srr rule, she cannot misreport to get matched. 
	We first show that agent $i$ cannot misreport to get matched to $u_c^1$.
	Each time an agent $j$ is added to $N_1$, it is because the agents in $N\setminus (N_1\cup\{j\})$ can be matched to $C_p$ to obtain a maximum beneficiary assignment. Since $i$ is not matched under truthful reporting, $i$ is not needed to obtain a maximum beneficiary assignment even if she reports all her eligible categories, which implies that $i$ is not needed to obtain a maximum beneficiary assignment if she reports a strict subset of her eligible categories. Hence, $i$ cannot affect the selection of agents preceding her in the baseline ordering that are added to $N_1$ and hence matched to $c_u^1$. Since $i$ was not matched to $u_c^1$, it means that when $i$ was considered to be added to $N_1$, the $q_{c_u^1}$ units of $c_u^1$ had already been used up. Therefore, $i$ could not have manipulated her priorities with respect to $C_p$ to get one of them.	
								
We have shown that $i$ cannot affect the set $N_1$, that is, which agents are matched to $c_u^1$. So we suppose that agents matched to $u_c^1$ are already fixed. 
Since the \rr rule is strategyproof, agent $i$ cannot get matched to a category in $C_p$ by misreporting. 

The remaining case is that, by misreporting, agent $i$ affects the set of agents who are not matched to a category in $C_p \cup \{c_u^1\}$ and, hence, compete with $i$ to be matched to $c_u^2$. 
We observe the following:
\begin{enumerate}
	\item Since agent $i$ is not matched to a category in $C_p$ and the \rr rule yields a maximum size matching, agent $i$ cannot affect the \emph{number} of agents matched to categories in $C_p$. 
	\item Since the \rr rule is weakly non-bossy, the set of agents \emph{lower} in the baseline ordering who compete to be matched to $c_u^2$ is unchanged under a misreport by $i$.
\end{enumerate}
The above two facts imply that under a misreport by $i$, the number of agents with a \emph{higher} baseline ordering than $i$ who are not matched to a category in $C_p \cup \{c_u^1\}$ and hence compete to be matched to $c_u^2$ does not change. Therefore, under any misreport, agent $i$ is not matched to $c_u^2$.
	 
	 \ref{item:srevnonbossy} Consider an agent $i$ who is unmatched under the \srr rule. 
	 By the proof of \ref{item:srevstrategyproof}, it follows that $i$ cannot affect 
										\begin{enumerate}
											\item the \emph{set} of agents who are matched to $c_u^1$,
											\item the \emph{number} of agents with higher baseline ordering who are not matched to a category in $C_p \cup \{c_u^1\}$, and
											\item the \emph{set} of agents with lower baseline ordering who are not matched to a category in $C_p \cup \{c_u^1\}$.
										\end{enumerate}
Under a truthful report, agent $i$ is unmatched and only agents with higher baseline ordering are matched to $c_u^2$. Hence, it follows that a misreport of agent $i$ does not affect the set of agents with lower baseline ordering who are matched to $c_u^2$.
	 
	 \ref{item:srevorderpreservation} Consider a matching $\mu$ returned by the \srr rule. Suppose it does not satisfy order preservation. 
	 Then there exist two agents $i,j\in N$ such that one of the following holds:
					\begin{enumerate}[label=(\roman*)]
						\item $\mu(i)=c_u^1$, $\mu(j)\neq c_u^1$, 
{$j\succ_{\mu(i)}i$}, and $i$ is eligible for category $\mu(j)$. 
						\item $\mu(j)=c_u^2$, $\mu(i)\in C_p\cup \{c_u^1\}$, 
{$j\succ_{\mu(i)}i$}, and $j$ is eligible for category $\mu(i)$.
						\end{enumerate}
						We first consider the violation of the first type: $\mu(i)=c_u^1$, $\mu(j)\neq c_u^1$, $j\succ_{\pi} i$, and $i$ is eligible for category $\mu(j)$. 
	We examine the step at which $i$ is considered when determining which agents get units from $c_u^1$.
	Since $\mu(i) = c_u^1$, agent $i$ is added to $N_1$.
	Thus, a maximum beneficiary assignment of the agents in $N\setminus(N_1\cup\{i\})$ exists.
	One such matching is $\mu$.
	The matching $\mu'$ obtained from $\mu$ by swapping the matches of $i$ and $j$ is a maximum beneficiary assignment (since $i$ is eligible for $\mu(j)$) for the agents in $N\setminus(N_1\cup\{j\})$.
	Since $j\succ_\pi i$ and $N_1$ weakly increases in every step, $N_1$ cannot have been larger when the algorithm considered agent $j$.
	Hence, at this earlier step, $\mu'$ was also a maximum beneficiary assignment for agents in $N\setminus(N_1\cup\{j\})$.
	But the existence of such a matching is the condition for adding $j$ to $N_1$, which contradicts that $\mu(j)\neq c_u^1$.

						Next we consider a violation of the second type: $\mu(j)=c_u^2$, $\mu(i)\in C_p\cup\{c_i^1\}$, $j\succ_{\pi} i$, and $j$ is eligible for category $\mu(i)$. 
						Since a violation of the first type cannot happen, we may assume that $\mu(i) \in C_p$. But since $j$ is not matched to a category in $C_p$, this implies that the \rr rule does not respect priorities, a contradiction.	
	 
	 \ref{item:srevpolytime} When agents are added iteratively to $N_1$, the algorithm requires checking if there exists a maximum beneficiary assignment of agents in $N\setminus (N_1\cup \{i\})$. 
	 This can be checked in polynomial time by algorithms for computing a maximum size b-matching. 
	 Once $N_1$ is fixed, we call the \rr rule, which we have already shown to be polynomial-time computable. 
	 After that the remaining units can be allocated in linear-time by going down the baseline ordering. 

\end{proof}

%
%
%
%

\begin{remark}[Soft reserves]
   We assumed that only matchings that satisfy the eligibility requirements are feasible. 
	The disadvantage of this approach is that some preferential category units may not be utilized even though some agents are unmatched. 
	If we do not impose eligibility requirements as hard constraints, the setting is referred to as the case of ``soft reserves''. 
	In that case, we can first compute a matching that complies with the eligibility requirements using the \srr rule. 
	If the resulting matching leaves some units from $C_p$ unassigned, we allocate those to unmachted agents in order of the baseline ordering.
	Assuming the preferential categories' priorities over \emph{ineligible} agents are consistent with the baseline ordering, the resulting rule satisfies all the properties in \Cref{thm:allaxioms:s-rev} except for ``hard'' compliance with the eligibility requirements. 
	The argument for strategyproofness is similar to the proof of \Cref{thm:allaxioms:s-rev}\ref{item:srevstrategyproof}. 
	(By misreporting, an agent cannot affect the number of agents with higher baseline order who are unmatched.)
\end{remark}

   		  \section*{Acknowledgements} 
		  
		  This material is based on work supported by the Deutsche Forschungsgemeinschaft under grant {BR~5969/1-1}.
		  The authors thank Adi Ganguly for helpful comments and Tayfun S{\"o}nmez  for additional pointers to the literature. 

\bibliographystyle{named}

\end{document}